\documentclass{article}
\usepackage{proceed2e}
\usepackage{natbib}
\usepackage{smallref}
\usepackage{amsmath}
\usepackage{amsfonts}
\usepackage{comment}
\usepackage{xcolor}
\usepackage{bm}
\usepackage[theorems]{mymacros}
\usepackage{graphicx}
\usepackage{url}
\usepackage{subfigure}
\usepackage{tikz}
\usetikzlibrary{arrows}
\usetikzlibrary{patterns}
\usetikzlibrary{decorations.pathmorphing}

\newcommand{\R}{{\mathcal R}}

\newcommand\intervene{\mathrm{do}}
\newcommand\DaDb[2]{\frac{d#1}{d#2}}
\newcommand\Gpa[2]{\mathrm{pa}_{#1}(#2)}

\tikzstyle{var}=[circle,draw=black,fill=white,thick,minimum size=20pt,inner sep=0pt]
\tikzstyle{arr}=[->,>=stealth',draw=black,line width=1pt]

\title{From Ordinary Differential Equations to Structural Causal Models: the deterministic case}
\author{{\bf Joris M.~Mooij} \\ Institute for Computing and \\ Information Sciences \\ Radboud University Nijmegen \\ The Netherlands
\And
{\bf Dominik Janzing} \\ Max Planck Institute \\ for Intelligent Systems \\ T\"ubingen, Germany
\And
{\bf Bernhard Sch\"olkopf} \\ Max Planck Institute \\ for Intelligent Systems \\ T\"ubingen, Germany
}

\begin{document}

\maketitle

\begin{abstract}
We show how, and under which conditions, the equilibrium states of a first-order Ordinary
Differential Equation (ODE) system can be described with a deterministic Structural Causal Model (SCM).
Our exposition sheds more light on the concept of causality as expressed within the framework of Structural Causal
Models, especially for cyclic models.
\end{abstract}

\section{Introduction}

Over the last few decades, a comprehensive theory for acyclic
causal models was developed (e.g., see \citep{Pearl00,SpiGlySch93}). In particular,
different, but related, approaches to causal inference and modeling have been proposed for the 
causally sufficient case. These approaches are based on different starting
points. One approach starts from the (local or global) causal Markov condition and links
observed independences to the causal graph. Another approach uses causal Bayesian networks
to link a particular factorization of the joint distribution of the variables to
causal semantics. The third approach uses a structural causal model (sometimes also
called structural equation model or functional causal model) where each effect
is expressed as a function of its direct causes and an unobserved noise variable. The relationships between these
aproaches are well understood \citep{Lauritzen96,Pearl00}.

Over the years, several attempts have been made to extend the theory to the
cyclic case, thereby enabling causal modeling of systems that involve feedback
\citep{Spirtes95,Koster96,PearlDechter96,Neal00,Hyttinen12}. 
However, the relationships between the different approaches mentioned
before do not immediately generalize to the cyclic case in general (although partial results
are known for the linear case and the discrete case). Nevertheless, several algorithms (starting 
from different assumptions) have been proposed for inferring cyclic causal models from
observational data
\citep{Richardson96,LacerdaSpirtesRamseyHoyer08,SchmidtMurphy09,Itani_etal10,Mooij_et_al_NIPS_11}.

The most straightforward extension to the cyclic case seems to be offered by
the structural causal model framework. Indeed, the formalism stays intact when one
simply drops the acyclicity constraint. However, the question then arises how to
interpret cyclic structural equations. One option is
to assume an underlying discrete-time dynamical system, in which the structural
equations are used as fixed point equations 
\citep{Spirtes95,Dash05,LacerdaSpirtesRamseyHoyer08,Mooij_et_al_NIPS_11,Hyttinen12},  
i.e., they are used as update rules to calculate the values at time
$t+1$ from the values at time $t$, and then one lets $t \to \infty$. 
Here we show how an alternative 
interpretation of structural causal models arises naturally when considering
systems of ordinary differential equations.
By considering how these differential equations behave in an equilibrium state, we arrive at a
structural causal model that is time independent, yet where the causal
semantics pertaining to interventions is still valid.
As opposed to the usual interpretation as discrete-time fixed point equations, the continuous-time dynamics is not defined by 
the structural equations. Instead, we describe how the structural equations arise from the given dynamics. 
Thus it becomes evident that different dynamics can yield identical structural causal models.
This interpretation sheds more light
on the meaning of structural equations, and does not make any substantial
distinction between the cyclic and acyclic cases.

It is sometimes argued that inferring causality amounts to simply inferring the
time structure connecting the observed variables, since the cause always
preceeds the effect. This, however, ignores two important facts: First, time
order between two variables does not tell us whether the earlier one caused the
later one, or whether both are due to a common cause. This paper addresses a second
counter argument: a variable need not necessarily refer to a measurement
performed at a certain time instance. Instead, a causal graph may formalize
how intervening on some variables influences the equilibrium state of others.
This describes a phenomenological level on which the original time structure
between variables gets lost, but causal graphs und structural equations may
still be well-defined. On this level, also cyclic structural equations get a
natural and well-defined meaning.

For simplicity, we consider only deterministic systems, and leave the extension
to stochastic systems with possible confounding as future work.

\section{Ordinary Differential Equations}
Let $\C{I} := \{1,\dots,D\}$ be an index set of variable labels. Consider variables $X_i \in \C{R}_i$ for $i \in \C{I}$, where
$\C{R}_i \subseteq \RN^{d_i}$. We will use normal font for a single variable and boldface for a tuple of variables
$\B{X}_I \in \prod_{i \in I} \C{R}_i$. 
\subsection{Observational system}
Consider a dynamical system $\C{D}$ described by $D$ coupled first-order ordinary 
differential equations and an initial condition $\B{X}_0 \in \R_{\C{I}}$:\footnote{We write $\dot{X} := \DaDb{X}{t}$.}
\begin{equation}\label{eq:ODE_system}\begin{split}
  \dot{X}_i(t) = f_i(\B{X}_{\Gpa{\C{D}}{i}}), \quad
  X_i(0) = (\B{X}_0)_i \quad \forall i \in \C{I}
\end{split}\end{equation}
Here, $\Gpa{\C{D}}{i} \subseteq \C{I}$ is the set of (indices of) \emph{parents}\footnote{Note that $X_i$ can be a parent of itself.} of 
variable $X_i$, and each $f_i : \R_{\Gpa{\C{D}}{i}} \to \R_i$ is a (sufficiently smooth) function. This dynamical system is assumed to 
describe the ``natural'' or ``observational'' 
state of the system, without any intervention from outside. We will assume that if $j \in \Gpa{\C{D}}{i}$,
then $f_i$ depends on $X_j$ (in other words, $f_i$ should not be constant when varying $X_j$).
Slightly abusing terminology, we will henceforth call such a dynamical system $\C{D}$ an Ordinary Differential Equation (ODE).

The \emph{structure} of these differential equations can be represented as a
directed graph $\C{G}_{\C{D}}$, with one node for each variable and a directed
edge from $X_i$ to $X_j$ if and only if $\dot X_j$ depends on $X_i$.

\subsubsection{Example: the Lotka-Volterra model}\label{sec:example_LotkaVolterra}

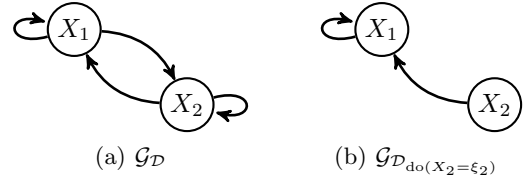
\begin{figure}[th]
  \centering
  \subfigure[$\C{G}_{\C{D}}$]{\begin{tikzpicture}
    \node[var] (X1) at (-0.5,1) {$X_1$};
    \node[var] (X2) at (1,0) {$X_2$};
    \draw[arr, bend left] (X1) edge (X2);
    \draw[arr, bend left] (X2) edge (X1);
    \draw[arr] (X1) edge[loop left] (X1);
    \draw[arr] (X2) edge[loop right] (X2);
  \end{tikzpicture}\label{fig:LotkaVolterra_obs_graph}}\qquad
  \subfigure[$\C{G}_{\C{D}_{\intervene(X_2=\xi_2)}}$]{\begin{tikzpicture}
    \node[var] (X1) at (-0.5,1) {$X_1$};
    \node[var] (X2) at (1,0) {$X_2$};
    \draw[arr, bend left] (X2) edge (X1);
    \draw[arr] (X1) edge[loop left] (X1);
  \end{tikzpicture}\label{fig:LotkaVolterra_int_graph}}
  \caption{\label{fig:LotkaVolterra_graph}(a) Graph of the Lotka-Volterra model \eref{eq:LotkaVolterraODE}; (b) Graph of
  the same ODE after the intervention $\intervene(X_2 = \xi_2)$, corresponding with \eref{eq:LotkaVolterraODE_intervened}.}
\end{figure}

The Lotka-Volterra model \citep{Murray02} is a well-known model from population biology, modeling the mutual
influence of the abundance of prey $X_1 \in [0,\infty)$ (e.g., rabbits) and the abundance of predators $X_2 \in [0,\infty)$ (e.g., wolves):
\begin{equation}\label{eq:LotkaVolterraODE}
  \begin{cases}
    \dot X_1 & = X_1 (\theta_{11} - \theta_{12} X_2) \\
    \dot X_2 & = -X_2 (\theta_{22} - \theta_{21} X_1)
  \end{cases}
  \qquad
  \begin{cases}
    X_1(0) = a \\
    X_2(0) = b
  \end{cases}
\end{equation}
with all parameters $\theta_{ij} > 0$ and initial condition satisfying $a \ge 0, b \ge 0$.
The graph of this system is depicted in Figure~\ref{fig:LotkaVolterra_obs_graph}.

\subsection{Perfect interventions}\label{sec:ODE_interventions}
\emph{Interventions} on the system $\C{D}$ described in \eref{eq:ODE_system} can be modeled in different ways. Here we will
focus on \emph{``perfect'' interventions}: for a subset $I \subseteq \C{I}$ of components,
we force the value of $\B{X}_I$ to attain some value $\B{\xi}_I \in \R_I$. In particular, we will assume that the 
intervention is active from $t=0$ to $t=\infty$, and that its value $\B{\xi}_I$ does not change over time. Inspired by the do-operator introduced by \citet{Pearl00}, we will denote this type of intervention as
$\intervene(\B{X}_I = \B{\xi}_{I})$. 

On the level of the ODE, there are many ways of realizing a given perfect intervention.
One possible way
is to add terms of the form $\kappa (\xi_i - X_i)$ (with $\kappa > 0$) to the expression for $\dot X_i$,
for all $i \in I$: 
\begin{equation}\label{eq:ODE_system_explicitly_intervened}\begin{split}
  \dot X_i(t) & = \begin{cases}
    f_i(\B{X}_{\Gpa{\C{D}}{i}}) + \kappa (\xi_i - X_i)        & i \in I \\
    f_i(\B{X}_{\Gpa{\C{D}}{i}}) & i \in \C{I} \setminus I,
  \end{cases}
  \\
  X_i(0) & = (\B{X}_0)_i 
\end{split}\end{equation}
This would correspond to extending the
system by components which monitor the values of $\{X_i\}_{i\in I}$ and exert 
negative feedback if they deviate from their target values $\{\xi_i\}_{i \in I}$.
Subsequently, we let $\kappa \to \infty$ to consider the idealized
situation in which the intervention completely overrides the other mechanisms that
normally determine the value of $\B{X}_I$.  
Assuming that the functions $\{f_i\}_{i \in I}$ are bounded, we can let $\kappa \to \infty$
and obtain the \emph{intervened system} $\C{D}_{\intervene(\B{X}_I = \B{\xi}_I)}$:
\begin{equation}\label{eq:ODE_system_intervened}\begin{split}
  \dot X_i(t) & = \begin{cases}
    0                          & \quad i \in I \\
    f_i(\B{X}_{\Gpa{\C{D}}{i}}) & \quad i \in \C{I} \setminus I,
  \end{cases}
  \\
  X_i(0) & = \begin{cases}
     \xi_i         & \quad i \in I \\
     (\B{X}_{0})_i & \quad i \in \C{I} \setminus I
  \end{cases}
\end{split}\end{equation}

A perfect intervention changes the graph $\C{G}_{\C{D}}$ associated to the ODE $\C{D}$ by removing the 
incoming arrows on the nodes corresponding to the intervened variables $\{X_i\}_{i\in I}$. 
It also changes the parent sets of intervened variables: for each $i \in I$, $\Gpa{\C{D}}{i}$ is replaced by $\Gpa{\C{D}_{\intervene(\B{X}_I = \B{\xi}_I})}{i} = \emptyset$.


\subsubsection{Example: Lotka-Volterra model}

Let us return to the example in section~\ref{sec:example_LotkaVolterra}. In this context,
consider the perfect intervention $\intervene(X_2 = \xi_2)$. This intervention
could be realized by monitoring the abundance of wolves very precisely and
making sure that the number equals the target value $\xi_2$ at all time (for example, by
killing an excess of wolves and introducing new wolves from some reservoir of wolves).
This leads to the following intervened ODE:
\begin{equation}\label{eq:LotkaVolterraODE_intervened}
  \begin{cases}
    \dot X_1 & = X_1 (\theta_{11} - \theta_{12} X_2) \\
    \dot X_2 & = 0
  \end{cases}
  \qquad
  \begin{cases}
    X_1(0) = a \\
    X_2(0) = \xi_2
  \end{cases}
\end{equation}
The corresponding intervened graph is illustrated in Figure~\ref{fig:LotkaVolterra_int_graph}.

\subsection{Stability}

An important concept in our context is \emph{stability}, defined as follows:
\begin{definition}\label{def:stability}
  The ODE $\C{D}$ specified in \eref{eq:ODE_system} is called \emph{stable} if there exists a unique equilibrium state 
  $\B{X}^* \in \R_{\C{I}}$ such that for any initial state $\B{X}_0 \in \R_{\C{I}}$, the system
  converges to this equilibrium state as $t \to \infty$:
  $$\exists!_{\B{X}^* \in \R_{\C{I}}}\ \forall_{\B{X}_0 \in \R_{\C{I}}} : \lim_{t\to\infty} \B{X}(t) = \B{X}^*.$$
\end{definition}
One can weaken the stability condition by demanding convergence to and uniqueness of the equilibrium only for a certain subset of all initial
states. For clarity of exposition, we will use this strong stability condition.

We can extend this concept of stability by considering a certain set of perfect interventions:
\begin{definition}\label{def:interventional_stability}
  Let $\C{J} \subseteq \C{P}(\C{I})$.\footnote{For a set $A$, we denote with $\C{P}(A)$ the power set of $A$ (the set of all subsets of $A$).}
  The ODE $\C{D}$ specified in \eref{eq:ODE_system} is called \emph{stable with respect to
  $\C{J}$} if for all $I \in \C{J}$ and for all $\B{\xi}_I \in \R_I$, the intervened ODE $\C{D}_{\intervene(\B{X}_I = \B{\xi}_I)}$
  has a unique equilibrium state
  $\B{X}^*_{\intervene(\B{X}_I = \B{\xi}_I)} \in \R_{\C{I}}$ such that for any initial state $\B{X}_0 \in \R_{\C{I}}$
  with $(\B{X}_0)_I = \B{\xi}_I$, the system converges to this equilibrium as $t \to \infty$:
  $$\exists!_{\B{X}^*_{\intervene(\B{X}_I = \B{\xi}_I)} \in \R_{\C{I}}}\ \forall_{\substack{\B{X}_0 \in \R_{\C{I}} \text{s.t.}\\(\B{X}_0)_I = \B{\xi}_I}} : \lim_{t\to\infty} \B{X}(t) = \B{X}^*_{\intervene(\B{X}_I = \B{\xi}_I)}.$$
\end{definition}
This definition can also be weakened by not demanding stability for all $\B{\xi}_I \in \R_I$, but for smaller subsets instead. 
Again, we will use this strong condition for clarity of exposition, although in a concrete example to be discussed later (see Section~\ref{sec:example_oscillators}),
we will actually weaken the stability assumption along these lines.


\subsubsection{Example: the Lotka-Volterra model}

The ODE \eref{eq:LotkaVolterraODE} of the Lotka-Volterra model is not stable, as discussed in detail 
by \cite{Murray02}. Indeed, it has two equilibrium states,
$(X_1^*,X_2^*) = (0,0)$ and $(X_1^*,X_2^*) = (\theta_{22} / \theta_{21}, \theta_{11} / \theta_{12})$.
The Jacobian of the dynamics is given by:
$$\nabla \B{f}(\B{X}) = \begin{pmatrix}
  \theta_{11} - \theta_{12} X_2 & -\theta_{12} X_1 \\
  \theta_{21} X_2               & -\theta_{22} + \theta_{21} X_1 \\
\end{pmatrix}$$
In the first equilibrium state, it has a positive and a negative eigenvalue ($\theta_{11}$
and $-\theta_{22}$, respectively), and hence this equilibrium is unstable. At the second equilibrium
state, the Jacobian simplifies to 
$$\begin{pmatrix}
  0 & -\theta_{12} \theta_{22} / \theta_{21} \\
  \theta_{21} \theta_{11} / \theta_{12} & 0 \\
\end{pmatrix}$$
which has two imaginary eigenvalues, $\pm i\sqrt{\theta_{11}\theta_{22}}$. 
One can show \citep{Murray02} that 
the steady state of the system is an undamped oscillation around this equilibrium.

The intervened system \eref{eq:LotkaVolterraODE_intervened} is only generically stable, i.e., for most values of 
$\xi_2$: the unique stable equilibrium state is $(X_1^*,X_2^*) = (0, \xi_2)$ as long as $\theta_{11} - \theta_{12} \xi_2 \ne 0$.
If $\theta_{11} - \theta_{12} \xi_2 = 0$, there exists a family of equilibria $(X_1^*,X_2^*) = (c, \xi_2)$ with $c \ge 0$. 

\subsubsection{Example: damped harmonic oscillators}\label{sec:example_oscillators}

The favorite toy example of physicists is a system of coupled harmonic oscillators.
Consider a one-dimensional system of $D$ point masses $m_i$ ($i=1,\dots,D$) with positions $Q_i \in \RN$ and momenta $P_i \in \RN$, coupled by springs
with spring constants $k_i$ and equilibrium lengths $l_i$, under influence of friction with friction coefficients $b_i$, with fixed end positions
(see also Figure~\ref{fig:MassSpring_system}). 

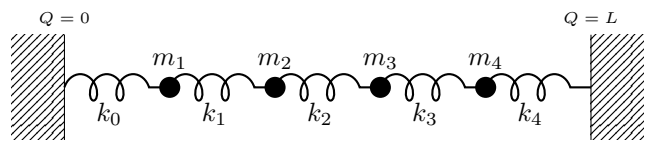
\begin{figure}[b]
  \vspace{-0.2cm}
  \centerline{\begin{tikzpicture}[scale=0.7]
    \fill[pattern=north east lines,draw=none] (-1,-1) rectangle (0,1);
    \draw (0,-1) -- (0,1);
    \fill[black] (2,0) circle (.2);
    \node (Q1) at (2,0.5) {$m_1$};
    \node (Q2) at (4,0.5) {$m_2$};
    \node (Q3) at (6,0.5) {$m_3$};
    \node (Q4) at (8,0.5) {$m_4$};
    \begin{scope}[xshift=-1.5mm]
      \node (k0) at (1,-0.5) {$k_0$};
      \node (k1) at (3,-0.5) {$k_1$};
      \node (k2) at (5,-0.5) {$k_2$};
      \node (k3) at (7,-0.5) {$k_3$};
      \node (k4) at (9,-0.5) {$k_4$};
    \end{scope}
    \fill[black] (4,0) circle (.2);
    \fill[black] (6,0) circle (.2);
    \fill[black] (8,0) circle (.2);
    \draw (10,-1) -- (10,1);
    \fill[pattern=north east lines,draw=none] (10,-1) rectangle (11,1);
    \draw[thick,decorate,decoration={coil,aspect=0.7,amplitude=5}] (0,0) -- (2,0);
    \draw[thick,decorate,decoration={coil,aspect=0.7,amplitude=5}] (2,0) -- (4,0);
    \draw[thick,decorate,decoration={coil,aspect=0.7,amplitude=5}] (4,0) -- (6,0);
    \draw[thick,decorate,decoration={coil,aspect=0.7,amplitude=5}] (6,0) -- (8,0);
    \draw[thick,decorate,decoration={coil,aspect=0.7,amplitude=5}] (8,0) -- (10,0);
    \node at (0,1.3) {\tiny $Q = 0$};
    \node at (10,1.3) {\tiny $Q = L$};
  \end{tikzpicture}}
  \caption{\label{fig:MassSpring_system}Mass-spring system for $D=4$.}
\end{figure}

We first sketch the 
qualitative behavior: there is a unique equilibrium position where the sum of forces vanishes for every single mass. Moving one or several masses
out of their equilibrium position stimulates vibrations of the entire system. Damped by friction, every mass converges to its unique and stable equilibrium position in the limit of infinite time. If one or several masses are fixed to positions different from their equilibrium points, the positions of the remaining masses still converge to unique (but different) equilibrium positions. 
The structural equations that we derive later will describe the change of the unconstrained equilibrium positions caused by fixing the others. 

The equations of motion for this system are given by:
\begin{align*}
  \dot{P_i} & = k_i (Q_{i+1} - Q_i - l_i) \\
            & \phantom{=} - k_{i-1} (Q_i - Q_{i-1} - l_{i-1}) - \frac{b_i}{m_i} P_i \\
  \dot{Q_i} &= P_i / m_i
\end{align*}
where we define $Q_0 := 0$ and $Q_{D+1} := L$.
The graph of this ODE is depicted in Figure~\ref{fig:MassSpring_graph}.
At equilibrium (for $t \to \infty$), all momenta
vanish, and the following equilibrium equations hold:
\begin{align*}
  0 & = k_i (Q_{i+1} - Q_i - l_i) - k_{i-1} (Q_i - Q_{i-1} - l_{i-1}) \\
  0 & = P_i
\end{align*}
which is a linear system of equations in terms of the $Q_i$. There are $D$ equations for $D$ unknowns $Q_1,\dots,Q_D$,
and one can easily check that it has a unique solution. 

A perfect intervention on $Q_i$ corresponds to fixating the position of the $i$'th mass.
Physically, this is achieved by adding a force that drives
$Q_i$ to some fixed location, i.e., the intervention on $Q_i$ is achieved through modifying the equation of 
motion for $\dot{P_i}$. To deal with this example in our framework, we consider the pairs $X_i := (Q_i,P_i) \in \RN^2$ 
to be the elementary variables.
Consider for example the perfect intervention $\intervene(X_2 = (\xi_2, 0))$, which effectively replaces the dynamical equations 
$\dot{Q_2}$ and $\dot{P_2}$ by $\dot{Q_2} = 0$, $\dot{P_2} = 0$ and their initial conditions by 
$(\B{Q}_0)_2 = \xi_2$, $(\B{P}_0)_2 = 0$. 
The graph of the corresponding ODE is depicted in Figure~\ref{fig:MassSpring_intervened_graph}.
Because of the friction, also this intervened system converges to a unique equilibrium that does not depend
on the initial value.

This holds more generally: for any perfect intervention on (any number) of pairs
$X_i$ of the type $\intervene(X_i = (\xi_i, 0))$, the intervened system will
converge towards a unique equilibrium because of the damping term. Interventions
that result in a nonzero value for any momentum $P_i$ while the corresponding position is fixed are physically impossible,
and hence will not be considered. Concluding, we have seen that the mass-spring system is 
stable with respect to perfect interventions on any number of position variables,
which we model mathematically as a joint intervention on the corresponding pairs
of position and momentum variables.

\begin{figure}[t]
  \centering
  \subfigure[$\C{G}_{\C{D}}$]{\begin{tikzpicture}
    \node[var] (Q1) at (0,0) {$Q_1$};
    \node[var] (Q2) at (2,0) {$Q_2$};
    \node[var] (Q3) at (4,0) {$Q_3$};
    \node[var] (Q4) at (6,0) {$Q_4$};
    \node[var] (P1) at (0,1.5) {$P_1$};
    \node[var] (P2) at (2,1.5) {$P_2$};
    \node[var] (P3) at (4,1.5) {$P_3$};
    \node[var] (P4) at (6,1.5) {$P_4$};
    \draw[dashed] (-0.5,-0.5) rectangle (1,2);
    \node at (0.7,-0.2) {$X_1$};
    \draw[dashed] (1.5,-0.5) rectangle (3,2);
    \node at (2.7,-0.2) {$X_2$};
    \draw[dashed] (3.5,-0.5) rectangle (5,2);
    \node at (4.7,-0.2) {$X_3$};
    \draw[dashed] (5.5,-0.5) rectangle (7,2);
    \node at (6.7,-0.2) {$X_4$};
    \draw[arr, bend left] (P1) edge (Q1);
    \draw[arr, bend left] (P2) edge (Q2);
    \draw[arr, bend left] (P3) edge (Q3);
    \draw[arr, bend left] (P4) edge (Q4);
    \draw[arr, bend left] (Q1) edge (P1);
    \draw[arr] (Q2) edge (P1);
    \draw[arr] (Q1) edge (P2);
    \draw[arr, bend left] (Q2) edge (P2);
    \draw[arr] (Q3) edge (P2);
    \draw[arr] (Q2) edge (P3);
    \draw[arr, bend left] (Q3) edge (P3);
    \draw[arr] (Q4) edge (P3);
    \draw[arr] (Q3) edge (P4);
    \draw[arr, bend left] (Q4) edge (P4);
  \end{tikzpicture}\label{fig:MassSpring_graph}}
  \subfigure[$\C{G}_{\C{D}_{\intervene(Q_2 = \xi_2, P_2 = 0)}}$]{\begin{tikzpicture}
    \node[var] (Q1) at (0,0) {$Q_1$};
    \node[var] (Q2) at (2,0) {$Q_2$};
    \node[var] (Q3) at (4,0) {$Q_3$};
    \node[var] (Q4) at (6,0) {$Q_4$};
    \node[var] (P1) at (0,1.5) {$P_1$};
    \node[var] (P2) at (2,1.5) {$P_2$};
    \node[var] (P3) at (4,1.5) {$P_3$};
    \node[var] (P4) at (6,1.5) {$P_4$};
    \draw[dashed] (-0.5,-0.5) rectangle (1,2);
    \node at (0.7,-0.2) {$X_1$};
    \draw[dashed] (1.5,-0.5) rectangle (3,2);
    \node at (2.7,-0.2) {$X_2$};
    \draw[dashed] (3.5,-0.5) rectangle (5,2);
    \node at (4.7,-0.2) {$X_3$};
    \draw[dashed] (5.5,-0.5) rectangle (7,2);
    \node at (6.7,-0.2) {$X_4$};
    \draw[arr, bend left] (P1) edge (Q1);
    \draw[arr, bend left] (P3) edge (Q3);
    \draw[arr, bend left] (P4) edge (Q4);
    \draw[arr, bend left] (Q1) edge (P1);
    \draw[arr] (Q2) edge (P1);
    \draw[arr] (Q2) edge (P3);
    \draw[arr, bend left] (Q3) edge (P3);
    \draw[arr] (Q4) edge (P3);
    \draw[arr] (Q3) edge (P4);
    \draw[arr, bend left] (Q4) edge (P4);
  \end{tikzpicture}\label{fig:MassSpring_intervened_graph}}
  \caption{Graphs of the dynamics of the mass-spring system for $D=4$. (a) Observational situation (b) Intervention $\intervene(Q_2 = \xi_2, P_2 = 0)$.}
\end{figure}
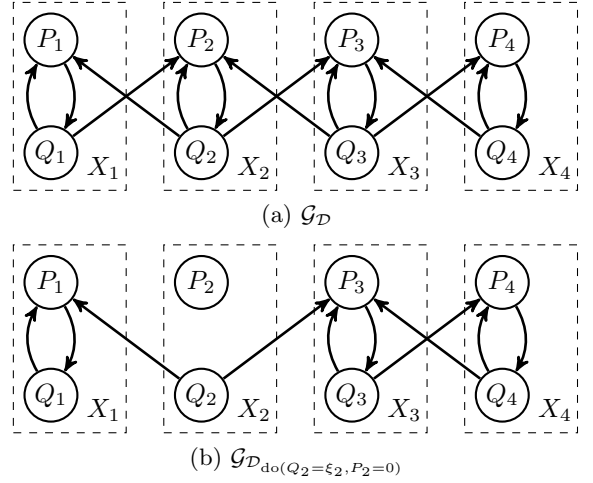

\section{Equilibrium equations}

In this section, we will study how the dynamical equations give rise to \emph{equilibrium equations} that
describe equilibrium states, and how these change under perfect interventions. This is an intermediate
representation on our way to structural causal models.

\subsection{Observational system}
At equilibrium, the rate of change of any variable is zero, by definition. Therefore, an equilibrium state of the
observational system $\C{D}$ defined in \eqref{eq:ODE_system} satisfies the following 
\emph{equilibrium equations}:
\begin{equation}\label{eq:ODE_equil_eqn}
  0 = f_i(\B{X}_{\Gpa{\C{D}}{i}}) \qquad \forall i \in \C{I}.
\end{equation}
This is a set of $D$ coupled equations with unknowns $X_1, \dots, X_D$. The stability assumption 
(cf.\ Definition~\ref{def:stability}) implies that there exists a unique solution $\B{X}^*$ of the equilibrium equations \eref{eq:ODE_equil_eqn}.

\subsection{Intervened systems}\label{sec:LEE_interventions}
Similarly, for the intervened system $\C{D}_{\intervene(\B{X}_i = \B{\xi}_i)}$ defined in \eqref{eq:ODE_system_intervened}, we obtain the following 
equilibrium equations:
\begin{equation}\label{eq:ODE_equil_eqn_intervened}
  \begin{cases} 
    0 = X_i - \xi_i \qquad & \forall i \in \C{I} \\
    0 = f_j(\B{X}_{\Gpa{\C{D}}{j}}) \qquad & \forall j \in \C{I} \setminus I
  \end{cases}
\end{equation}
If the system is stable with respect to this intervention (cf.\ Definition~\ref{def:interventional_stability}), then there exists a unique solution 
$\B{X}^*_{\intervene(\B{X}_I=\B{\xi}_I)}$ of the intervened equilibrium equations \eref{eq:ODE_equil_eqn_intervened}.

Note that we can also go directly from the equilibrium equations \eref{eq:ODE_equil_eqn} of the observational system
$\C{D}$ to the equilibrium equations \eref{eq:ODE_equil_eqn_intervened} of the intervened system $\C{D}_{\intervene(\B{X}_I = \B{\xi}_I)}$,
simply by replacing the equilibrium equations ``$0 = f_i(\B{X}_{\Gpa{\C{D}}{i}})$'' for $i \in I$ by 
equations of the form ``$0 = X_i - \xi_i$''.
Indeed, note that the modified dynamical equation
$$\dot X_i = f_i(\B{X}_{\Gpa{\C{D}}{i}}) + \kappa (\xi_i - X_i)$$
yields an equilibrium equation of the form
$$0 = f_i(\B{X}_{\Gpa{\C{D}}{i}}) + \kappa (\xi_i - X_i)$$
which, in the limit $\kappa \to \infty$, reduces to
$0 = X_i - \xi_i$ (assuming that $f_i$ is bounded).
This seemingly trivial observation will turn out to be quite important.

\subsection{Labeling equilibrium equations}\label{sec:ODE_to_LEE}

If we would consider the equilibrium equations as a set of \emph{unlabeled}
equations $\{\C{E}_i : i \in \C{I}\}$, where $\C{E}_i$ denotes the equilibrium
equation ``$0 = f_i(\B{X}_{\Gpa{\C{D}}{i}})$'' (or ``$0 = X_i - \xi_i$'' after an
intervention) for $i \in \C{I}$, then we will not be able to
correctly predict the result of interventions, as we do not know \emph{which}
of the equilibrium equations should be changed in order to model the particular
intervention. This information is present in the dynamical system
$\C{D}$ (indeed, the terms ``$\dot X_i$'' in the l.h.s.\ of the
dynamical equations in \eref{eq:ODE_system} indicate the targets of the intervention), but is lost when
considering the corresponding equilibrium equations \eref{eq:ODE_equil_eqn} 
as an unlabeled set (because the terms ``$\dot X_i$'' are all replaced by zeroes).

This important information can be preserved by labeling the equilibrium equations.
Indeed, the \emph{labeled} set of equilibrium equations $\C{E} := \{(i, \C{E}_i) : i \in \C{I}\}$
contains all information needed to predict how equilibrium states
change on arbitrary (perfect) interventions. Under an
intervention $\intervene(\B{X}_I = \B{\xi}_I)$, the equilibrium equations are changed as follows:
for each intervened component $i \in I$, the equilibrium equation $\C{E}_i$ is
replaced by the equation $\C{\tilde E}_i$ defined as ``$0 = X_i - \xi_i$'', whereas the
other equilibrium equations $\C{E}_j$ for $j \in \C{I} \setminus I$ do not change.
Assuming that the dynamical system is stable with respect to this intervention, this
modified system of equilibrium equations describes the new equilibrium obtained under
the intervention. We conclude that the information about the values of equilibrium 
states and how these change under perfect interventions is encoded in the 
labeled equilibrium equations.

\subsection{Labeled equilibrium equations}

The previous considerations motivate the following formal
definition of a system of Labeled Equilibrium Equations (LEE) and
their semantics under interventions.
\begin{definition}
  A system of \emph{Labeled Equilibrium Equations (LEE)} $\C{E}$ for $D$ variables
  $\{X_i\}_{i\in \C{I}}$ with $\C{I} := \{1,\dots,D\}$ consists of $D$ \emph{labeled equations}
  of the form
    \begin{equation}\label{eq:labeled_equations}
      \C{E}_i: \quad 0 = g_i(\B{X}_{\Gpa{\C{E}}{i}}), \qquad i \in \C{I},
    \end{equation}
  where $\Gpa{\C{E}}{i} \subseteq \C{I}$ is the set of (indices of) \emph{parents}
  of variable $X_i$, and each $g_i : \R_{\Gpa{\C{E}}{i}} \to \R_i$ is a function.
\end{definition}
The \emph{structure} of an LEE $\C{E}$ can be represented as a directed graph 
$\C{G}_{\C{E}}$, with one node for each variable and a directed
edge from $X_i$ to $X_j$ (with $j \ne i$) if and only if $\C{E}_i$ depends on $X_j$.

A perfect intervention transforms an LEE into another (intervened) LEE:
\begin{definition}
Let $I \subseteq \C{I}$ and $\B{\xi}_I \in \R_I$. For the perfect intervention $\intervene(\B{X}_I = \B{\xi}_I)$ that forces the variables 
$\B{X}_I$ to take the value $\B{\xi}_I$, the intervened LEE 
$\C{E}_{\intervene(\B{X}_I=\B{\xi}_I)}$ is obtained by replacing the
labeled equations of the original LEE $\C{E}$ by the following modified labeled equations:
  \begin{equation}\label{eq:labeled_equations_intervened}
  0 = \begin{cases}
    X_i - \xi_i & i \in I \\
    g_i(\B{X}_{\Gpa{\C{E}}{i}}) & i \in \C{I} \setminus I.
  \end{cases}\end{equation}
\end{definition}

We define the concept of solvability for LEEs that mirrors the definition of stability for ODEs:
\begin{definition}
An LEE $\C{E}$ is called \emph{solvable} if there exists a unique solution $\B{X}^*$ to the system of (labeled) equations $\{\C{E}_i\}$.
An LEE $\C{E}$ is called \emph{solvable with respect to $\C{J} \subseteq \C{P}(\C{I})$} if
for all $I \in \C{J}$ and for all $\B{\xi}_I \in \R_I$, the intervened LEE $\C{E}_{\intervene(\B{X}_I = \B{\xi}_I)}$ 
is solvable.
\end{definition}

As we saw in the previous section, an ODE induces an LEE in a straightforward way.
The graph $\C{G}_{\C{E}_{\C{D}}}$ of the induced LEE $\C{E}_{\C{D}}$ is equal to the graph $\C{G}_{\C{D}}$ of the ODE $\C{D}$.
It is immediate that if the ODE $\C{D}$ is stable, then the induced LEE $\C{E}_{\C{D}}$ is solvable.
As we saw at the end of Section~\ref{sec:LEE_interventions}, our ways of modeling interventions on
ODEs and on LEEs are compatible. We will spell out this important result in detail.
\begin{theorem}\label{theo:ODE_to_LEE}
Let $\C{D}$ be an ODE, $I \subseteq \C{I}$ and $\B{\xi}_I \in \R_I$. (i) Applying the perfect intervention
$\intervene(\B{X}_I = \B{\xi}_I)$ to the induced LEE $\C{E}_{\C{D}}$ gives the same result as constructing
the LEE corresponding to the intervened ODE $\C{D}_{\intervene(\B{X}_I = \B{\xi}_I)}$:
$$(\C{E}_{\C{D}})_{\intervene(\B{X}_I = \B{\xi}_I)} = \C{E}_{\C{D}_{\intervene(\B{X}_I = \B{\xi}_I)}}.$$ 
(ii) Stability of the ODE $\C{D}$ with respect to the intervention $\intervene(\B{X}_I = \B{\xi}_I)$
implies solvability of the induced intervened LEE $\C{E}_{\C{D}_{\intervene(\B{X}_I = \B{\xi}_I)}}$,
and the corresponding equilibrium and solution $\B{X}^*_{\intervene(\B{X}_I = \B{\xi}_I)}$ are identical.\qed
\end{theorem}

\subsection{Example: damped harmonic oscillators}

Consider again the example of the damped, coupled harmonic oscillators of section~\ref{sec:example_oscillators}.
The labeled equilibrium equations are given explicitly by:
\begin{equation}\label{eq:LEE_oscillators}
  \C{E}_i: \qquad \begin{cases}
    0 & = k_i (Q_{i+1} - Q_i - l_i) \\
      & \phantom{=} - k_{i-1} (Q_i - Q_{i-1} - l_{i-1}) \\
    0 & = P_i
  \end{cases}
\end{equation}

\section{Structural Causal Models}

In this section we will show how an LEE representation can be mapped to the more popular
representation of Structural Causal Models, also known as Structural Equation Models \citep{Bollen89}. We follow the terminology of 
\citet{Pearl00}, but consider here only the subclass of \emph{deterministic} SCMs.

\subsection{Observational}

The following definition is a special case of the more general definition in \citep[Section 1.4.1]{Pearl00}:
\begin{definition}
  A \emph{deterministic Structural Causal Model (SCM)} $\C{M}$ on $D$ variables
  $\{X_i\}_{i\in \C{I}}$ with $\C{I} := \{1,\dots,D\}$ consists of $D$ 
  \emph{structural equations} of the form
    \begin{equation}\label{eq:structural_equations}
    X_i = h_i(\B{X}_{\Gpa{\C{M}}{i}}), \qquad i \in \C{I},
    \end{equation}
  where $\Gpa{\C{M}}{i} \subseteq \C{I} \setminus \{i\}$ is the set of (indices of) \emph{parents}
  of variable $X_i$, and each $h_i : \R_{\Gpa{\C{M}}{i}} \to \R_i$ is a function.
\end{definition}
Each structural equation contains a function $h_i$ that depends on
the components of $\B{X}$ in $\Gpa{\C{M}}{i}$. We think of the parents $\Gpa{\C{M}}{i}$ as the
\emph{direct causes} of $X_i$ (relative to $\B{X}_{\C{I}}$) and the function $h_i$ as
the \emph{causal mechanism} that maps the direct causes to the effect $X_i$.
Note that the l.h.s.\ of a structural equation by definition contains only 
$X_i$, and that the r.h.s.\ is a function of variables \emph{excluding}
$X_i$ itself. In other words, $X_i$ is not considered to be a direct cause
of itself. The \emph{structure} of an SCM $\C{M}$ is often represented as a 
directed graph $\C{G}_{\C{M}}$, with one node for each variable and a directed
edge from $X_i$ to $X_j$ (with $j \ne i$) if and only if $h_i$ depends on $X_j$.
Note that this graph does not contain ``self-loops'' (edges pointing from a node to itself), by definition.

\subsection{Interventions}\label{sec:SCM_interventions}

A Structural Causal Model $\C{M}$ comes with a specific semantics for modeling perfect
interventions \citep{Pearl00}:
\begin{definition}
Let $I \subseteq \C{I}$ and $\B{\xi}_I \in \R_I$. For the perfect intervention $\intervene(\B{X}_I = \B{\xi}_I)$ that forces the variables 
$\B{X}_I$ to take the value $\B{\xi}_I$, the intervened SCM 
$\C{M}_{\intervene(\B{X}_I=\B{\xi}_I)}$ is obtained by replacing the
structural equations of the original SCM $\C{M}$ by the following modified structural equations:
  \begin{equation}\label{eq:structural_equations_intervened}
  X_i = \begin{cases}
    \xi_i & i \in I \\
    h_i(\B{X}_{\Gpa{\C{M}}{i}}) & i \in \C{I} \setminus I.
  \end{cases}\end{equation}
\end{definition}
The reason that the equations in a SCM are called ``structural equations'' (instead
of simply ``equations'') is that they also contain information for modeling 
interventions, in a similar way as the labeled equilibrium equations contain this 
information. In particular, the l.h.s.\ of the structural equations indicate the 
targets of an intervention.\footnote{In \citet{Pearl00}'s words: ``Mathematically, the distinction between structural and algebraic equations is that
the latter are characterized by the set of solutions to the entire system of
equations, whereas the former are characterized by the solutions of each
individual equation. The implication is that any subset of structural equations is,
in itself, a valid model of reality---one that prevails under some set of interventions.''}

\subsection{Solvability}
Similarly to our definition for LEEs, we define:
\begin{definition}
An SCM $\C{M}$ is called \emph{solvable} if there exists a unique solution $\B{X}^*$ to the system
of structural equations. An SCM $\C{M}$ is called \emph{solvable with respect to $\C{J} \subseteq \C{P}(\C{I})$} if
for all $I \in \C{J}$ and for all $\B{\xi}_I \in \R_I$, the intervened SCM $\C{M}_{\intervene(\B{X}_I = \B{\xi}_I)}$ 
is solvable.
\end{definition}
Note that each (deterministic) SCM $\C{M}$ with acyclic graph $\C{G}_{\C{M}}$ is solvable, even with
respect to the set of all possible intervention targets, $\C{P}(\C{I})$. This is not necessarily true if
directed cycles are present.

\subsection{From labeled equilibrium equations to deterministic SCMs}\label{sec:LEE_to_SCM}

Finally, we will now show that under certain stability assumptions on an ODE $\C{D}$, we can represent the information
about (intervened) equilibrium states that is contained in the corresponding set of labeled equilibrium equations $\C{E}_{\C{D}}$ as an SCM
$\C{M}_{\C{E}_{\C{D}}}$.

First, given an LEE $\C{E}$, we will construct an induced SCM $\C{M}_{\C{E}}$, provided certain
solvability conditions hold:
\begin{definition}
If
for each $i \in \C{I}$, the LEE $\C{E}$ is solvable with respect to some $I_i \subseteq \C{I}$
with $\Gpa{\C{E}}{i} \setminus \{i\} \subseteq I_i \subseteq \C{I} \setminus \{i\}$, then
it is called \emph{structurally solvable}.
\end{definition}
If the LEE $\C{E}$ is structurally solvable, we can proceed as follows. Let $i \in \C{I}$. We define the induced parent set $\Gpa{\C{M}_{\C{E}}}{i} := \Gpa{\C{E}}{i} \setminus \{i\}$. 
Assuming structural solvability of $\C{E}$, under the perfect intervention 
$\intervene(\B{X}_{I_i} = \B{\xi}_{I_i})$, there is a unique solution 
$\B{X}^*_{\intervene(\B{X}_{I_i} = \B{\xi}_{I_i})}$ to the intervened LEE, for any value of 
$\B{\xi}_{I_i} \in \R_{I_i}$. This defines a function $h_i : \R_{\Gpa{\C{M}_{\C{E}}}{i}} \to \R_i$ given by the $i$'th
component $h_i(\B{\xi}_{\Gpa{\C{M}_{\C{E}}}{i}}) := \big(\B{X}^*_{\intervene(\B{X}_{I_i} = \B{\xi}_{I_i})}\big)_i$.
The $i$'th structural equation of the induced SCM $\C{M}_{\C{E}}$ is defined as:
$$X_i = h_i(\B{X}_{\Gpa{\C{M}_{\C{E}}}{i}}).$$
Note that this equation is \emph{equivalent} to the labeled equation $\C{E}_i$ in the sense that they
have identical solution sets $\{(X_i^*, \B{X}^*_{\Gpa{\C{M}_{\C{E}}}{i}})\}$. Repeating this procedure for all $i \in \C{I}$, we 
obtain the induced SCM $\C{M}_{\C{E}}$.

This construction is designed to preserve the important mathematical structure. In particular:
\begin{lemma}\label{lemm:LEE_to_SCM}
Let $\C{E}$ be an LEE, $I \subseteq \C{I}$ and $\B{\xi}_I \in \R_I$
and consider the perfect intervention $\intervene(\B{X}_I = \B{\xi}_I)$. 
Suppose that both the LEE $\C{E}$ and the intervened LEE $\C{E}_{\intervene(\B{X}_I = \B{\xi}_I)}$ are structurally solvable.
(i) Applying the intervention $\intervene(\B{X}_I = \B{\xi}_I)$ to the induced SCM $\C{M}_{\C{E}}$ gives the same result as constructing
the SCM corresponding to the intervened LEE $\C{E}_{\intervene(\B{X}_I = \B{\xi}_I)}$:
$$(\C{M}_{\C{E}})_{\intervene(\B{X}_I = \B{\xi}_I)} = \C{M}_{\C{E}_{\intervene(\B{X}_I = \B{\xi}_I)}}.$$ 
(ii) Solvability of the LEE $\C{E}$ with respect to the intervention $\intervene(\B{X}_I = \B{\xi}_I)$
implies solvability of the induced SCM $\C{M}_{\C{E}}$ with respect to the same intervention $\intervene(\B{X}_I = \B{\xi}_I)$,
and their respective solutions $\B{X}^*_{\intervene(\B{X}_I = \B{\xi}_I)}$ are identical.
\end{lemma}
\begin{proof}
The first statement directly follows from the construction of the induced SCM. The key observation regarding solvability is the following.
From the construction above it directly follows that
\begin{equation*}\begin{split}
  & \forall_{\B{X}_{\Gpa{\C{E}}{i}} \in \R_{\Gpa{\C{E}}{i}}}: \\ 
  & \qquad 0 = g_i(\B{X}_{\Gpa{\C{E}}{i}}) \iff X_i = h_i(\B{X}_{\Gpa{\C{E}}{i} \setminus \{i\}}).
\end{split}\end{equation*}
This trivially implies that
\begin{equation*}
  \forall_{\B{X} \in \R_{\C{I}}}: 0 = g_i(\B{X}_{\Gpa{\C{E}}{i}}) \iff X_i = h_i(\B{X}_{\Gpa{\C{M}_{\C{E}}}{i}}).
\end{equation*}
This implies that each simultaneous solution of all labeled equations is a simultaneous
solution of all structural equations, and vice versa:
\begin{equation*}\begin{split}
  \forall_{\B{X} \in \R_{\C{I}}}: \qquad & \Big( \big[\forall_{i \in \C{I}}: 0 = g_i(\B{X}_{\Gpa{\C{E}}{i}})\big] \\
                                          & \phantom{\Big(} \iff \big[\forall_{i \in \C{I}}: X_i = h_i(\B{X}_{\Gpa{\C{M}_{\C{E}}}{i}})\big]\Big).
\end{split}\end{equation*}
The crucial point is that this still holds if an intervention replaces some of the equations (by $0 = X_i - \xi_i$
and $X_i = \xi_i$, respectively, for all $i \in I$).
\end{proof}

\subsection{From ODEs to deterministic SCMs}
We can now combine all the results and definitions so far to construct a deterministic
SCM from an ODE under certain stability conditions. We define:
\begin{definition}
An ODE $\C{D}$ is called \emph{structurally stable} if 
for each $i \in \C{I}$, the ODE
$\C{D}$ is stable with respect to some $I_i \subseteq \C{I}$ with
$\Gpa{\C{D}}{i} \setminus \{i\} \subseteq I_i \subseteq \C{I} \setminus \{i\}$.
\end{definition}
Consider the diagram in Figure~\ref{fig:commutative_diagram}.
Here, the labels of the arrows correspond with the numbers of the sections that discuss 
the corresponding mapping. The downward mappings correspond with a particular intervention
$\intervene(\B{X}_I = \B{\xi}_I)$, applied at the different levels (ODE, induced LEE, 
induced SCM). Our main result:
\begin{theorem}\label{theo:ODE_to_LEE_to_SCM}
If both the ODE $\C{D}$ and the intervened ODE $\C{D}_{\intervene(\B{X}_I = \B{\xi}_I)}$
are structurally stable, then:
(i) The diagram in Figure~\ref{fig:commutative_diagram} commutes.\footnote{This means
that it does not matter in which direction one follows the arrows, the end result will be the same.}
(ii) If furthermore, the ODE $\C{D}$ is stable with respect to the intervention $\intervene(\B{X}_I = \B{\xi}_I)$, the SCM $\C{M}_{\C{E}_{\C{D}_{\intervene(\B{X}_I = \B{\xi}_I)}}}$ has a unique solution that coincides with the stable equilibrium of the intervened ODE $\C{D}_{\intervene(\B{X}_I = \B{\xi}_I)}$.
\end{theorem}
\begin{proof}
Immediate from Theorem~\ref{theo:ODE_to_LEE} and Lemma~\ref{lemm:LEE_to_SCM}.
\end{proof}

\begin{figure*}
\begin{center}
  \begin{tikzpicture}
  \node [text width=30mm, text centered, draw] (ode) at (0,0) {ODE\\ $\C{D}$};
  \node [text width=30mm, text centered, draw] (lee) at (5,0) {LEE\\ $\C{E}_{\C{D}}$};
  \node [text width=30mm, text centered, draw] (scm) at (10,0) {SCM\\ $\C{M}_{\C{E}_{\C{D}}}$};
  \draw [|->,shorten >=2pt,shorten <=2pt] (ode) -- (lee) node [above,midway] {\ref{sec:ODE_to_LEE}};
  \draw [|->,dashed,shorten >=2pt,shorten <=2pt] (lee) -- (scm) node [above,midway] {\ref{sec:LEE_to_SCM}};
  \node [text width=30mm, text centered, draw] (intode) at (0,-2) {intervened ODE\\ $\C{D}_{\intervene(\B{X}_I = \B{\xi}_I)}$};
  \node [text width=30mm, text centered, draw] (intlee) at (5,-2) {intervened LEE\\ $\C{E}_{\C{D}_{\intervene(\B{X}_I = \B{\xi}_I)}}$};
  \node [text width=30mm, text centered, draw] (intscm) at (10,-2) {intervened SCM\\ $\C{M}_{\C{E}_{\C{D}_{\intervene(\B{X}_I = \B{\xi}_I)}}}$};
  \draw [|->,shorten >=2pt,shorten <=2pt] (intode) -- (intlee) node [above,midway] {\ref{sec:ODE_to_LEE}};
  \draw [|->,dashed,shorten >=2pt,shorten <=2pt] (intlee) -- (intscm) node [above,midway] {\ref{sec:LEE_to_SCM}};
  \draw [|->,shorten >=2pt,shorten <=2pt] (ode) -- (intode) node [right,midway] {\ref{sec:ODE_interventions}};
  \draw [|->,shorten >=2pt,shorten <=2pt] (lee) -- (intlee) node [right,midway] {\ref{sec:LEE_interventions}};
  \draw [|->,shorten >=2pt,shorten <=2pt] (scm) -- (intscm) node [right,midway] {\ref{sec:SCM_interventions}};;
\end{tikzpicture}
\end{center}
\caption{\label{fig:commutative_diagram}Each of the arrows in the
diagram corresponds with a mapping that is described in the section that the 
label refers to. The dashed arrows are only defined under structural
solvability assumptions on the LEE. If the ODE $\C{D}$ and intervened ODE $\C{D}_{\intervene(\B{X}_I = \B{\xi}_I)}$ are structurally stable, this diagram commutes (cf.\ Theorem~\ref{theo:ODE_to_LEE_to_SCM}).}
\end{figure*}
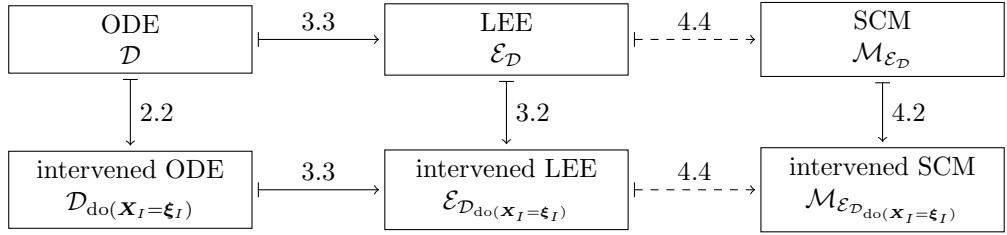

Note that even though the ODE may contain self-loops (i.e., the time
derivative $\dot X_i$ could depend on $X_i$ itself, and hence $i \in \Gpa{\C{D}}{i}$), 
the induced SCM $\C{M}_{\C{E}_{\C{D}}}$ does \emph{not} contain self-loops by construction 
(i.e., $i \not\in \Gpa{\C{M}_{\C{E}_{\C{D}}}}{i}$).
Somewhat surprisingly, the structural stability conditions actually imply the
existence of self-loops (because if $X_i$ would not occur in the equilibrium
equation $(\C{E}_{\C{D}})_i$, its value would be undetermined and hence the 
equilibrium would not be unique).

Whether one prefers the SCM representation over the LEE representation is
mainly a matter of practical considerations: both representations
contain all the necessary information to predict the results of arbitrary
perfect interventions, and one can easily go from the LEE representation to the
SCM representation. One can also easily go in the opposite direction, but this
cannot be done in a unique way. For example, one could rewrite each structural
equation $X_i = h_i(\B{X}_{\Gpa{\C{M}}{i}})$ as the equilibrium equation 
$0 = h_i(\B{X}_{\Gpa{\C{M}}{i}}) - X_i$, but also as the equilibrium equation 
$0 = h_i^3(\B{X}_{\Gpa{\C{M}}{i}}) - X_i^3$ (in both cases, it would be given the label $i$).

In case the dynamics contains no directed cycles (not considering
self-loops), the advantage of the SCM representation is that it is 
more explicit. Starting at the variables without parents, and following the topological
ordering of the corresponding directed acyclic graph, we directly obtain the solution of an SCM
by simple substitution in a finite number of steps. On the other hand, the LEE representation is more implicit,
and we need to solve a set of equations. 
In the cyclic case, one needs to solve a set of equations in both representations,
and the difference is merely cosmetical. However, one could argue that the LEE representation
is slightly more natural in the cyclic case, as it does not force us to make additional (structural) stability assumptions.

\subsection{Example: damped harmonic oscillators}

Figure~\ref{fig:MassSpring_collapsed_graph} shows the graph of the structural causal model induced
by our construction. It reflects the intuition that at equilibrium, (the position of) each mass has a direct causal 
influence on (the positions of) its neighbors.
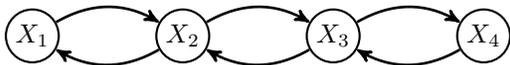
\begin{figure}[b]
  \centerline{\begin{tikzpicture}
    \node[var] (X1) at (0,1.5) {$X_1$};
    \node[var] (X2) at (2,1.5) {$X_2$};
    \node[var] (X3) at (4,1.5) {$X_3$};
    \node[var] (X4) at (6,1.5) {$X_4$};
    \draw[arr, bend left] (X2) edge (X1);
    \draw[arr, bend left] (X1) edge (X2);
    \draw[arr, bend left] (X3) edge (X2);
    \draw[arr, bend left] (X2) edge (X3);
    \draw[arr, bend left] (X4) edge (X3);
    \draw[arr, bend left] (X3) edge (X4);
  \end{tikzpicture}}
  \caption{\label{fig:MassSpring_collapsed_graph}Graph of the structural causal model induced by the mass-spring system for $D=4$.}
\end{figure}
Observing that the momentum variables always vanish at equilibrium (even for any perfect intervention that we consider),
we can decide that the only relevant variables for the SCM are the position variables $Q_i$. Then, we end up with the following 
structural equations:
\begin{equation}\label{eq:SCM_oscillators}
  Q_i = \frac{k_i (Q_{i+1} - l_i) + k_{i-1} (Q_{i-1} + l_{i-1})}{k_i + k_{i+1}}.
\end{equation}






\section{Discussion}

In many empirical sciences (physics, chemistry, biology, etc.) and in engineering, 
differential equations are a commonly used modeling tool. When estimating system 
characteristics from data, they are especially useful if
measurements can be done on the relevant time scale. If equilibration time
scales become too small with respect to the temporal resolution of measurements,
however, the more natural representation may be in terms of structural causal
models. The main contribution of this work is to build an explicit bridge from
the world of differential equations to the world of causal models.
Our hope is that this may aid in broadening the impact of causal modeling.

Note that information is lost when going from a dynamical system representation to an
equilibrium representation (either LEE or SCM), in particular the rate of convergence
toward equilibrium. If time-series data is
available, the most natural representation may be the dynamical system representation.
If only snapshot data or equilibrium data is available, the dynamical system representation
can be considered to be overly complicated, and one may use the LEE or SCM representation
instead.

We have shown one particular way in which structural causal models can be ``derived''.
We do not claim that this is the only way: on the contrary, SCMs can probably be 
obtained in several other ways and from other representations as well. One issue that
we have not yet addressed is that of \emph{constants of motion}. For example, if we
would not fix the end points of the chain of harmonic oscillators, then the total
momentum of the system would depend on the initial condition, and therefore the dynamics
would not be stable anymore according to the definition we have used here. We believe
that these and similar issues can probably be solved by being more explicit about which
variables in the dynamics will become part of the structural causal model. We plan to
address this in future work.


We also intend to extend the basic framework described here towards the more general 
stochastic case. Uncertainty or ``noise'' can enter in two different ways: via uncertainty
about certain (constant) parameters of the differential equations, and via latent variables.
A complicating factor that has to be addressed then (which does not play a role in the deterministic
case) is confounding. 




\subsubsection*{Acknowledgements} 

We thank Bram Thijssen, Tom Claassen, Tom Heskes and Tjeerd Dijkstra for stimulating discussions. JM was supported by NWO, the
Netherlands Organization for Scientific Research (VENI grant 639.031.036).
 
\bibliographystyle{apalike}
\bibliography{dyneq_to_scm,bibfile3}
\end{document}